\documentclass[1p]{elsarticle}
\usepackage{hyperref}
\journal{}
\usepackage{fix-cm}
\usepackage{amsmath,amsthm}
\usepackage{latexsym,enumerate,amssymb,amsbsy}
\usepackage{mathtools}
\usepackage{graphicx}
\usepackage{graphics,color}
\usepackage{verbatim}
\usepackage{url}
\usepackage{dblfloatfix}
\usepackage{pgf,tikz}
\usetikzlibrary{trees,automata,arrows,shapes.geometric}
\usepackage{booktabs}
\usepackage{xfrac}
\usepackage{amsfonts}
\usepackage{stmaryrd}
\usepackage{algorithm}
\usepackage[noend]{algpseudocode}
\usepackage{url}
\usepackage{colortbl}
\newtheorem{proposition}{Proposition}
\newtheorem{theorem}[proposition]{Theorem}
\newtheorem{corollary}[proposition]{Corollary}
\newtheorem{lemma}[proposition]{Lemma}

\newtheorem{example}{Example}

\newcommand{\Mod}[1]{\ (\textup{mod}\ #1)}

\newtheorem{definition}{Definition}

\DeclareMathOperator{\per}{per}
\DeclareMathOperator{\ord}{ord}

\DeclarePairedDelimiter{\floor}{\lfloor}{\rfloor}

\usepackage{listings}
\usepackage{xcolor}

\definecolor{codegreen}{rgb}{0,0.6,0}
\definecolor{codegray}{rgb}{0.5,0.5,0.5}
\definecolor{codepurple}{rgb}{0.58,0,0.82}
\definecolor{backcolour}{rgb}{0.95,0.95,0.92}

\lstdefinestyle{mystyle}{
	backgroundcolor=\color{backcolour},
	commentstyle=\color{codegreen},
	keywordstyle=\color{magenta},
	numberstyle=\tiny\color{codegray},
	stringstyle=\color{codepurple},
	basicstyle=\ttfamily\footnotesize,
	breakatwhitespace=false,
	breaklines=true,
	captionpos=b,
	keepspaces=true,
	numbers=left,
	numbersep=5pt,
	showspaces=false,
	showstringspaces=false,
	showtabs=false,
	tabsize=2
}

\begin{document}
	
	\begin{frontmatter}
		
		\title{A New Approach to Determine the Minimal Polynomials of Binary Modified de Bruijn Sequences}
		
		\cortext[cor1]{Corresponding author}
		\author[ugm,uny]{Musthofa\corref{cor1}}
		\ead{musthofa2019@mail.ugm.ac.id, musthofa@uny.ac.id}
		
		\author[ugm]{Indah Emilia Wijayanti}
		\ead{ind\_wijayanti@ugm.ac.id}
		
		\author[ugm]{Diah Junia Eksi Palupi}
		\ead{diah\_yunia@ugm.ac.id}
		
		\author[ntu]{Martianus Frederic Ezerman}
		\ead{fredezerman@ntu.edu.sg}
		
		\address[ugm]{Department of Mathematics, 
			Faculty of Mathematics and Natural Sciences,\\Universitas Gadjah Mada, Sekip Utara BLS 21, Yogyakarta 55281, Indonesia.}
		
		\address[uny]{Department of Mathematics Education, 
			Universitas Negeri Yogyakarta,\\
			1 Colombo Road, Yogyakarta 55281, Indonesia.}
		
		\address[ntu]{School of Physical and Mathematical Sciences, Nanyang Technological University,\\
			21 Nanyang Link, Singapore 637371.}
		
		\begin{abstract}
			A binary modified de Bruijn sequence is an infinite and periodic binary sequence derived by removing a zero from the longest run of zeros in a binary de Bruijn sequence. The minimal polynomial of the modified sequence is its unique least-degree characteristic polynomial. Leveraging on a recent characterization, we devise a novel general approach to determine the minimal polynomial. We translate the characterization into a problem of identifying a Hamiltonian cycle in a specially constructed graph. Along the way, we demonstrate the usefullness of computational tools from the cycle joining method in the modified setup.
		\end{abstract}
		
		
	\end{frontmatter}
	
	
	\section{Introduction}\label{sec:intro}
	Given a positive integer $n$, a {\it binary de Bruijn sequence of order $n$} is an infinite sequence over $\mathbb{F}_2$ with period $2^n$. Each binary $n$-tuple appears exactly once per period. Much has been done in the studies of such sequences. One can start from a primitive polynomial $p(x) \in \mathbb{F}_2[x]$ with $\deg(p(x))=n$. The linear feedback shift register (LFSR) whose characteristic polynomial is $p(x)$ produces a maximum length sequence $\mathbf{m}$, also known as an $m$-sequence, of period $2^n-1$. Appending a $0$ to the longest run of zeroes in $\mathbf{m}$ results in a de Bruijn sequence. 
	
	We know from \cite{Bruijn46}, which independently rediscovered the result in \cite{StMarie1894}, that the number of binary de Bruijn sequences of order $n$ is $2^{2^{n-1}-n}$. The number of primitive polynomial over $\mathbb{F}_2$ is $n^{-1} \, \varphi (2^n-1)$, where $\varphi(\cdot)$ is the Euler totient function. The number of de Bruijn sequences from the set of all primitive polynomials becomes miniscule compared to $2^{2^{n-1}-n}$ as $n$ grows. 
	
	There are other methods than the route via $m$-sequences. Interested readers may want to consult Fredricksen's survey \cite{Fred82} and more recent works, such as the ones by Chang, Ezerman, Ling, and Wang in \cite{Chang2019}, by Gabric, Sawada, Williams, and Wong in~\cite{Gabric20}, and in many of their respective references. 
	
	The sequence $\mathbf{m}$ produced by a primitive polynomial $p(x)$ can also be seen as a modification of the corresponding de Bruijn sequence by removing a $0$ from the longest run of zeros. We call a sequence modified from any de Bruijn sequence by such removal of a $0$ a {\it modified de Bruijn sequence}. 
	
	A measure of predictability of a sequence is given by its {\it linear span} or {\it linear complexity}. It is defined to be the degree of the shortest linear recursion or minimal polynomial that produces the sequence. The higher the span is the less predictable the sequence becomes. For cryptographic applications, for example, sequences with large spans are desirable. 
	
	The respective minimal and maximal values for the linear span of de Bruijn sequences are $2^{n-1}+n$ and $2^n-1$, for $n \geq 3$. This fact and further results on the distribution of the linear span values can be found in \cite{Etzion99}. The extremal values were initially established as bounds in~\cite{Chan82}, where the upper bound was then shown to be achievable. Etzion and Lempel showed how to construct de Bruijn sequences of minimal span in~\cite{Etzion1984}. 
	
	Let $\mathbf{s}$ denote a modified binary de Bruijn sequence of order $n$ and period $2^n-1$. Let $\widetilde{\mathbf{s}}$ be its corresponding de Bruijn sequence. One often prefers using $\mathbf{s}$, instead of $\widetilde{\mathbf{s}}$, since the presence of the all zero string of length $n$ can be undesirable. On the other hand, there are instances when $\widetilde{\mathbf{s}}$ has an optimal linear complexity but $\mathbf{s}$ performs poorly in this measure.
	
	\begin{example}
		The sequence $\widetilde{\mathbf{s}}:=(0000100110101111)$, with $n=4$, has linear complexity $15$. The linear complexity of $\mathbf{s}=(000100110101111)$ drops significantly to $4$. 
	\end{example}
	
	Much less is known regarding the minimal polynomials of modified de Bruijn sequences. An early work on this topic was done by Mayhew and Golomb in~\cite{Mayhew1990}.
	Subsequent works by Kyureghyan in \cite{Kyureghyan2008}, by Tan, Xu, and Qi in \cite{Tan2018}, and a more recent one by Wang, Cheng, Wang, and Qi in \cite{Wang2020} have not managed to supply any systematic method to determine the minimal polynomials.
	
	For any $n \geq 4$, let $\Omega$ be the set of all nonzero polynomials of degree $< n$, that is,
	\begin{equation}\label{eq:omega}
		\Omega := \{0 \neq h(x) : h(x) \in \mathbb{F}_2[x] \mbox{ with } \deg(h(x)) < n\}.
	\end{equation}
	In this work, we propose a general method to design the minimal polynomial of a modified de Bruijn sequence. To briefly outline our method, we defer the formal definition of terms related to polynomials to Section \ref{sec:prelims}. We pick up from where the work by Tan {\it et al.} in~\cite{Tan2018} ends. In particular, we rely on the following useful characterization.
	
	\begin{theorem}\cite[Theorem 3.2]{Tan2018}\label{teotan3}
		Let $m$ and $n$ be positive integers satisfying $m \geq n > 1$. A polynomial $f(x) \in \mathbb{F}_2[x]$ of degree $m$ is the minimal polynomial of a modified de Bruijn sequence of order $n$ if and only if the following conditions hold:	
		\begin{enumerate}
			\item The polynomial $f(x)$ satisfies 
			\begin{equation}\label{eq:on_f}
				f(0) = 1 \mbox{ and its period is } \per(f(x)) = 2^n-1.
			\end{equation}
			\item Let $f^*(x)$ be the reciprocal polynomial of $f(x)$. There exist $g(x) \in \mathbb{F}_2[x]$ with $\deg(g(x)) = m-n$ such that
			\begin{align}
				& g(0) = 1, \label{eq:g_1}\\
				& \gcd(g(x),f^*(x) ) = 1 \mbox{, and } \label{eq:fg_1}\\
				& \Omega = \{(g(x) \, x^i \Mod{f^*(x)}) \Mod{x^n} :  0 \leq i < 2^n-1\}.\label{eq:fg_2}
			\end{align}
		\end{enumerate}
	\end{theorem}
	Theorem~\ref{teotan3} requires the two polynomials $f(x)$ and $g(x)$ not only to satisfy their respective conditions in (\ref{eq:on_f}) and (\ref{eq:g_1}), but also to simultaneously meet the requirements in (\ref{eq:fg_1}) and (\ref{eq:fg_2}) with $\Omega$ as defined in (\ref{eq:omega}). 
	
	Our first step is to posit an auxiliary self-reciprocal polynomial 
	\begin{equation}\label{eq:main}
		F(x) := \sum_{i=0}^{2^n-2} x^i = 1 + x + \ldots + x^{2^n-3}+ x^{2^n-2}
	\end{equation}
	by taking the product of \emph{all} elements in the set 
	\begin{equation}\label{eq:setA}
		\mathcal{A}:= \bigcup_{\substack{\delta \mid n \\ \delta \neq 1}} 
		\{\mbox{irreducible } r(x) \in \mathbb{F}_2[x] : \deg(r(x))=\delta\}.
	\end{equation}
	Since the product of all irreducible polynomials whose degrees divide $n$ is $x^{2^n}+x$ and the only irreducible polynomials of degree $1$ are $x$ and $x+1$, it is clear that 
	\[
	F(x)= \prod_{r(x) \in \mathcal{A}} r(x) = \frac{x^{2^n}+x}{x \, (x+1)}. 
	\]
	
	The next move is to transform the problem of finding a suitable $g(x)$ into a problem of determining a Hamiltonian cycle $\mathcal{H}$ in a specially constructed graph $\Gamma_n$. We give a systematic way to determine $g(x)$. The details are to be covered in Sections \ref{sec:poly2graph} and \ref{sec:hamiltonian}. 
	
	\medskip
	
	\noindent
	{\bf Our main contributions}
	\begin{enumerate}
		\item For a given $n \geq 3$, we construct a graph $\Gamma_n$ with the property that every Hamiltonian cycle $\mathcal{H}$ in $\Gamma_n$ corresponds to a modified de Bruijn sequence. We then propose basic algorithms to identify numerous Hamiltonian cycles in $\Gamma_n$.
		\item To each Hamiltonian cycle $\mathcal{H}$ in $\Gamma_n$ we associate a canonical generator polynomial $c_{\mathcal{H}}(x)$. Once $c_{\mathcal{H}}(x)$ is found, we compute $d(x):= \gcd(c_{\mathcal{H}}(x),F(x))$ and prove that the minimal polynomial $m_{\mathbf{s}}(x)$ of $\mathbf{s}$ is $f^{*}(x)$, where $\displaystyle{f(x):=\frac{F(x)}{d(x)}}$.
		\item We supply basic computational tools as proofs of concept.
	\end{enumerate}
	
	\section{Preliminaries}\label{sec:prelims}
	
	Let $\mathbf{s}:=s_0,s_1,\ldots$ be an infinite sequence over a given finite field $\mathbb{F}_q$. If there is a positive integer $N$ for which $N$ is the smallest number such that $s_i = s_{i+N}$ for all $i \geq 0$, then $\mathbf{s}$ is an $N$-{\it periodic} sequence and we write $\mathbf{s}=(s_0,s_1,\ldots,s_{N-1})$. Such an $N$ is the {\it period} $\per(\mathbf{s})$ of $\mathbf{s}$. The sum of two infinite sequences $\mathbf{s}:=s_0,s_1,\ldots$ and $\mathbf{t}:=t_0,t_1,\ldots$ over the same finite field is $\mathbf{s}+\mathbf{t} =s_0+t_0, s_1+t_1, \ldots$ and the scalar multiple $c \, \mathbf{s}$ with $c \in \mathbb{F}_q$ is simply $c \, s_0, c \, s_1,\ldots$.  Henceforth, unless otherwise stated, all sequences in this work are binary, that is, $q=2$. 
	
	Let $L$ be the {\it (left) shift operator} that sends 
	\begin{equation}\label{eq:opL}
		\mathbf{s}=(s_0,s_1,\ldots,s_{N-1}) \mapsto L(\mathbf{s})=(s_1,\ldots,s_{N-1},s_0).
	\end{equation}
	By convention $L^0$ fixes the sequence. Two sequences $\mathbf{a}$ and $\mathbf{b}$ are called {\it distinct} or {\it (cyclically) inequivalent} if one is \emph{not} the cyclic shift of the other, that is, there is no integer $k \geq 0$ such that $\mathbf{a} = L^k \mathbf{b}$.
	
	A monic polynomial $f(x)$ in the ring of binary polynomials $\mathbb{F}_2[x]$ of indeterminate $x$ is a {\it characteristic polynomial} of $\mathbf{s}=(s_0,s_1,\ldots,s_{N-1})$ if $f(L) (\mathbf{s}) =(0,0,\ldots,0)$. One can then call $\mathbf{s}$ a {\it linear feedback shift register} (LFSR) sequence. As an LFSR sequence, $\mathbf{s}$ may have many characteristic polynomials. We identify the unique characteristic polynomial $m_{\mathbf{s}}(x)$ of least degree as its {\it minimal polynomial}. Any characteristic polynomial of $\mathbf{s}$ is divisible by $m_{\mathbf{s}}(x)$. 
	
	The minimal polynomial of a sequence gives the sequence's {\it measure of predictability}. The degree $\Delta:=\Delta_{\mathbf{s}}=\deg(m_{\mathbf{s}}(x))$ is the {\it linear complexity} or the {\it linear span} of $\mathbf{s}$. Knowing any $\Delta$-tuple in $\mathbf{s}$ allows us to reconstruct $\mathbf{s}$ completely. The zero sequence has linear span $0$.
	
	The {\it reciprocal polynomial} of 
	$a(x):=1 + a_1 \, x + \ldots + a_{n-1} \,x^{n-1} + x^n $ is the polynomial 
	\begin{equation}\label{eq:reciprocal}
		a^*(x):= x^n a(x^{-1}) = 1 + a_{n-1} \, x + \ldots + a_1 \, x^{n-1} + x^n. 
	\end{equation}
	A polynomial is {\it self-reciprocal} if it is its own reciprocal. The {\it order} of $a(x)$, denoted by $\ord(a(x))$, is the least positive integer $\lambda$ for which $a(x)$ divides $x^{\lambda}-1$. The minimal polynomial of an $N$-periodic sequence has order $N$. 
	
	We will use the {\it rational fraction representation} of any $\mathbf{s}=(s_0,s_1,\ldots,s_{N-1})$. Further details can be found in \cite[Chapter 6 Section 3]{LN97}. The {\it generating function} of $\mathbf{s}$ is the element 
	\begin{equation}\label{eq:genfunc}
		s(x):= s_0 + s_1 \, x + \ldots = \sum_{i=0}^{\infty} s_i \, x^i
	\end{equation}
	in the {\it ring of formal power series} over $\mathbb{F}_2$. Since $\mathbf{s}$ is periodic, it can be represented as a rational function
	\begin{equation}\label{eq:ratio}
		s(x)= \sum_{i=0}^{\infty} s_i \, x^i = \frac{g(x)}{f(x)},
	\end{equation}
	with $f(x)$ and $g(x)$ in $\mathbb{F}_2[x]$ such that 
	\begin{equation}\label{eq:fracpoly}
		\deg(g(x)) < \deg(f(x)), ~ \gcd(g(x),f(x))=1, ~ f(0) = 1.
	\end{equation}
	The minimal polynomial $m_{\mathbf{s}}(x)$ of $\mathbf{s}$ is the reciprocal $f^*(x)$ of the denominator $f(x)$ in (\ref{eq:ratio}). The converse also holds. For any $f(x)$ and $g(x)$ in $\mathbb{F}_2[x]$ satisfying (\ref{eq:fracpoly}), there exists a periodic sequence $\mathbf{s}$ for which (\ref{eq:ratio}) holds. 
	
	The requirement that $\gcd(g(x),f(x))=1$ is not strictly necessary. Indeed, there will be occasions in the sequel that we relax this condition and allow for a rational function representation $s(x)$ with $\gcd(g(x),f(x)) = d(x) \neq 1$. The context will make it clear whether the representation is the simplest one or the more relaxed version.
	
	\begin{theorem}\label{thm:alpha}\cite[Section 2]{Chan82}
		If $\mathbf{a}$ is binary de Bruijn sequence, then the minimal polynomial of $\mathbf{a}$ has the form $a(x)=(x+1)^{z}$ for some integer $z$ satisfying $2^{n-1}+1 \leq z \leq 2^n$.
	\end{theorem}
	
	Let $\mathcal{I}(n)$ denote the number of binary irreducible polynomials of degree $n$ in $\mathbb{F}_2[x]$. Let $\mu(n)$ be the M{\"o}bius function. Gauss' general formula~\cite[Theorem 3.25]{LN97} says that 
	$\displaystyle{\mathcal{I}(n) = \frac{1}{n} \sum_{j \mid n} \mu(j) \,  2^{\frac{n}{j}}}$. Sequence A001037 in~\cite{OEIS} lists $\mathcal{I}(n)$.
	
	\begin{theorem}\cite[Theorem 2]{Mayhew1990}
		Let $\widetilde{\mathbf{s}}$ be a de Bruijn sequence of order $n \geq 4$ whose modified sequence is $\mathbf{s}$. Then 
		\begin{equation}\label{eq:possible}
			\Delta(\mathbf{s}) = \sum_{d \mid n} a_d \cdot d \mbox{, with } 0 \leq a_d \leq \mathcal{I}(d).
		\end{equation}
	\end{theorem}
	
	Let $\mathbf{s}$ be a modified de Bruijn sequence of order $n$. For small values of $n$, it is known that $n \leq \Delta_{\mathbf{s}} \leq 2^n-2$. It is conjectured in~\cite{Tan2018} that $\Delta_{\mathbf{s}} \notin \{n+1,n+2,\ldots, 3n-1\}$. If $\mathbf{s}$ is not an $m$-sequence, then we know from \cite[Corollary 3]{Wang2020} that $\Delta_{\mathbf{s}} > \frac{5}{4} n$.
	
	We now recall useful results on the rational fraction representations of modified de Bruijn sequences over $\mathbb{F}_q$ established by Tan {\it et al.} in \cite{Tan2018}. For any nonnegative integer $k$, the \textit{$k$-shifted sequence} of $\mathbf{a} = 
	(a_0,a_1,\ldots,a_{N-1} )$ is 
	\begin{equation}\label{eq:kshifted}
		L^k \mathbf{a} := 
		\left(a_k, a_{k+1}, \ldots, a_{N-1},a_0, \ldots, a_{k-1} \right).
	\end{equation}
	The set of all shifted sequences of $\mathbf{a}$ is 
	$\left\{L^k \mathbf{a} : 0 \leq k < N \right\}$.
	
	\begin{lemma}\cite[Lemma 3.5]{Tan2018}\label{teotan1}
		Let $\mathbf{a}=(a_0, a_1, \ldots, a_{N-1})$ be a given $N$-periodic sequence with rational fraction representation $\displaystyle{\frac{g(x)}{f(x)}}$. Then, for any $ 0 \leq k < N$, the rational fraction representation of its $k$-shifted sequence $L^k \mathbf{a}$ is
		\begin{equation}\label{eq:rational}
			\frac{g_k(x)}{f(x)} \mbox{, where } 
			g_k(x) := \left(g(x) \, x^{N-k} \right) \Mod{f(x)}.
		\end{equation}
	\end{lemma}
	
	\begin{lemma}\cite[Lemma 3.6]{Tan2018}\label{teotan2}
		Let $\mathbf{a} = (a_0,a_1,\ldots,a_{2^n-2})$ be a sequence of period $2^n-1$ with  $\displaystyle{\frac{g(x)}{f(x)}}$ as its rational fraction representation. Then $\mathbf{a}$ is a modified de Bruijn sequence of order $n$ if and only if 
		\[
		\Omega =\{ (g(x) \, x^i \Mod{f(x)}) \Mod{x^n} : 0 \leq i < 2^n-1\}.
		\]
	\end{lemma}
	Lemma~\ref{teotan2} asserts that $\mathbf{a}$ is a modified de Bruijn sequence if and only if $g_k(x) \Mod{x^n}$ traverses all nonzero polynomials of degree less than $n$ as $k$ goes from $0$ to $2^n-2$. The preparation in Lemmas \ref{teotan1} and \ref{teotan2} leads to Theorem~\ref{teotan3} in the introduction.
	
	We use graph theoretic notions commonly defined in standard textbooks. A {\it directed gaph} or a {\it digraph} is an ordered pair $G:=(V,E)$ where $V$ is a set of vertices and $E$ a set of {\it ordered} pairs called {\it directed edges} or {\it arcs}. In this work, a digraph does not have multiple arcs on the same ordered pair of vertices, altough it may contain a loop. A {\it Hamiltonian path (cycle)} is a path (cycle) that visits every vertex exactly once, with each arc traced according to its direction. We will often use the terms Hamiltonian cycle and Hamiltonian path interchangably, without causing ambiguity or losing generality.

	\section{From Polynomials to Directed Graphs}\label{sec:poly2graph}
	From the work of Mayhew and Golomb in \cite{Mayhew1990} we know that the minimal polynomial of modified binary de Bruijn sequences of order $n$ is a product of distinct irreducible polynomials of degree $d \neq 1$, with $d \mid n$. Setting aside the $m$-sequences built from primitive polynomials of degree $n$, there had not been any systematic way to determine the minimal polynomial of a given modified de Bruijn sequence. 
	
	Not all possible values given in (\ref{eq:possible}) are in fact the actual values of the linear span. For $n=5$, for instance, there is no modified binary de Bruijn sequence with minimal polynomial $f(x)=(x^5+x^2+1)(x^5+x^3+1)$ although this degree $10$ polynomial is a product of distinct irreducible polynomials of degrees dividing $5$. The degrees taken by the actual minimal polynomials for $n \in \{4,5,6\}$ are listed in Table~\ref{table:degree}. 
	
	\begin{table}[h!]
		\caption{The degrees of actual minimal polynomials of modified de Bruijn sequences of order $n \in \{4,5,6\}$, reproduced from \cite{Mayhew1990}.}\label{table:degree}
		\centering
		\begin{tabular}{c l}
			\toprule
			$n$ & Degrees of actual minimal polynomials \\
			\midrule
			$4$ & $4, 12, 14$\\
			$5$ & $5, 15, 20, 25, 30 $\\
			$6$ & $6, 27, 30, 32, 33, 35, 36, 38, 39, 41, 42, 44, 
			45, 47, 48, 50, 51, 53, 54, 56, 57, 59, 60, 62$\\
			\bottomrule		
		\end{tabular}
	\end{table}
	
	\begin{example}
		For $n=4$ there are exactly $10$ modified de Bruijn sequences having the maximal linear span $14$. Their minimal polynomial is $F(x) = 1+x+x^3+\ldots+x^{14}$. The $10$ polynomials $g(x)$ that, each, satisfies the requirements in Theorem~\ref{teotan3} with $F(x)$ taking the place of $f(x)$, are given in Table~\ref{tab:n4}. 
		Performing long division, we easily confirm that the first entry in Table~\ref{tab:n4} with $g(x)=x^{10} + x^8 + x^5 + x + 1$ has a representation
		\[
		\frac{g(x)}{F(x)} = \left(1 + x^2 + x^5 + x^6 + x^8 + x^9 + x^{10} + x^{11}\right) + x^{15} 
		\left(\frac{g(x)}{F(x)}\right),
		\]
		corresponding to the modified de Bruijn sequence $(1,0,1,0,0,1,1,0,1,1,1,1,0,0,0)$. The rest of the entries can be similarly interpreted.
	\end{example}
	
	\begin{table}[h!]
		\caption{Degree $10$ polynomials $g(x) = \sum_{j=0}^{10} g_j x^j $, written as $g_{10} \,g_9 \,\ldots \,g_1 \, g_0$, that generate $\Omega$ and their corresponding sequences.}\label{tab:n4}
		\centering
		\resizebox{\textwidth}{!}{%
			\begin{tabular}{cc | cc}
				\toprule
				$g(x)$ & Modified de Bruijn Sequence $\mathbf{s}$&
				$g(x)$ & Modified de Bruijn Sequence$\mathbf{s}$\\
				\midrule
				$10100100011$ & $(1,0,1,0,0,1,1,0,1,1,1,1,0,0,0)$ &
				$11011000101 $ & $(1, 1, 1, 1, 0, 0, 1, 0, 1, 1, 0, 1,0,0,0)$ \\
				
				$10011010111$ & $(1,0,0,1,1,1,1,0,1,0,1,1,0,0,0)$ &
				$10100011011 $ & $(1, 0, 1, 1, 0, 1, 0, 0, 1, 1, 1, 1,0,0,0)$ \\
				
				$11101011001$ & $(1,1,0,1,0,1,1,1,1,0,0,1,0,0,0)$ &
				$11010111011 $ & $(1, 0, 1, 1, 0, 0, 1, 1, 1, 1, 0, 1,0,0,0)$ \\
				
				$10001101011$ & $(1,0,1,1,1,1,0,1,0,0,1,1,0,0,0)$ &
				$11011101011 $ & $(1, 0, 1, 1, 1, 1, 0, 0, 1, 1, 0, 1,0,0,0)$ \\
				
				$11010110001$ & $(1, 1, 0, 0, 1, 0, 1, 1, 1, 1, 0, 1,0,0,0)$ &
				$11000100101$& $(1,1,1,1,0,1,1,0,0,1,0,1,0,0,0)$ \\
				\bottomrule		
			\end{tabular}
		}
	\end{table}
	
	We define a digraph (directed graph) $\Gamma_n(V,E)$, or simply $\Gamma$ when $n$ is clear from the context, based on the set $\Omega$ in (\ref{eq:omega}) as follows. We associate each nonzero polynomial 
	\[
	a(x) = a_0 + a_1 \,x + \ldots + a_{n-2} \, x^{n-2} + a_{n-1} \, x^{n-1} \in \Omega
	\]
	with the $n$-string 
	\[
	\mathbf{a} := a_{n-1},a_{n-2},\ldots,a_{1},a_0
	\]
	and its integer representation 
	\[
	A:= a_{n-1} \, 2^{n-1} + a_{n-2} \, 2^{n-2} +  \ldots + a_1 \, 2 + a_0.
	\]
	Hence, there is a one-to-one correspondence between elements in $\Omega$ and the integers in $\{1,2,\ldots,2^{n}-1\}$, which we use as the vertex set $V$. Let $a(x), \, b(x) \in \Omega$ be seen as vertices $A,B \in V$. We add {\it an arc from $A$ to $B$} if and only if
	\begin{align}
		b (x) &= x \, a(x) \Mod{x^n} \mbox{ or } \label{eq:double}\\
		b (x) &=x \, a(x) \Mod{x^n} + \sum_{i=0}^{n-1} x^i. \label{eq:complement}
	\end{align}
	
	The arc governed by (\ref{eq:double}) is from $A$ to $B:=2A \Mod{(2^n-1)}$ while the one defined by (\ref{eq:complement}) is from $A$ to $B:= (2^n-1) - (2A \Mod{(2^n-1)})$. We call the former the {\it doubling arc}, marked in blue and labelled by a $0$, and the latter the {\it double-then-complement arc}, marked in red and labelled by a $1$. For brevity, the names are abbreviated to {\it double} and {\it complement} arcs. 
	
	The outdegree of each vertex is $2$, except for the vertex $2^{n-1}$ whose outdegree is $1$ since $0 \notin \Omega$. This vertext has only a red arc to its complement vertex $2^n-1$. Each vertex has indegree $2$, except for the vertex $2^n-1$ whose only inbound edge comes from $2^{n-1}$. There is a loop from vertex $A$ to itself if and only if $3A = 2^n-1$. This vertex is clearly unique. The graph $\Gamma_n$ is simple for all $n$ such that $3 \nmid (2^n-1)$. 
	
	\begin{figure}[ht!]
		\centering
		\begin{tikzpicture}
			[
			> = stealth,
			shorten > = 1pt,
			auto,
			node distance = 1.9cm,
			semithick
			]
			
			\tikzstyle{every state}=
			\node[rectangle,fill=white,draw,rounded corners,minimum size = 4mm]
			
			\node[state,fill=lightgray] (1) {$1$};
			\node[state,fill=lightgray] (2) [right of=1] {$2$};
			\node[state,fill=lightgray] (3) [right of=2] {$3$};
			\node[state,fill=lightgray] (4) [right of=3] {$4$};
			\node[state,fill=lightgray] (5) [right of=4] {$5$};
			
			\node[state,fill=lightgray] (6) [below of=1] {$6$};
			\node[state,fill=lightgray] (7) [right of=6] {$7$};
			\node[state,fill=lightgray] (8) [right of=7] {$8$};
			\node[state,fill=lightgray] (9) [right of=8] {$9$};
			\node[state,fill=lightgray] (10) [right of=9] {$10$};
			
			\node[state,fill=lightgray] (11) [below of=6] {$11$};
			\node[state,fill=lightgray] (12) [right of=11] {$12$};
			\node[state,fill=lightgray] (13) [right of=12] {$13$};
			\node[state,fill=lightgray] (14) [right of=13] {$14$};
			\node[state,fill=lightgray] (15) [right of=14] {$15$};
			
			\path[->,blue] (1) edge node [above] {$0$} (2);
			\path[->,red] (1) edge[bend right=20] node [below] {$1$} (13);
			
			\path[->,blue] (2) edge[bend left=20] node [above] {$0$} (4);
			\path[->,red] (2) edge node [below] {$1$} (11);
			
			\path[->,blue] (3) edge[bend right=15] node [below] {$0$} (6);
			\path[->,red] (3) edge node [above] {$1$} (9);
			
			\path[->,blue] (4) edge node [below] {$0$} (8);
			\path[->,red] (4) edge node [above] {$1$} (7);
			
			\path[->,blue] (5) edge node [right] {$0$} (10);
			\path[->,red] (5) edge[loop left] node [above] {$1$} (5);
			
			\path[->,blue] (6) edge[bend right=10] node [below] {$0$} (12);
			\path[->,red] (6) edge[bend right=15] node [above] {$1$} (3);
			
			\path[->,blue] (7) edge[bend left=5] node [below] {$0$} (14);
			\path[->,red] (7) edge node [left] {$1$} (1);
			
			\path[->,red] (8) edge[bend left=5] node [right] {$1$}(15);
			
			\path[->,blue] (9) edge node [below] {$0$}(2);
			\path[->,red] (9) edge node [left] {$1$}(13);
			
			\path[->,blue] (10) edge node [below] {$0$}(4);
			\path[->,red] (10) edge[bend right=3] node [right] {$1$}(11);
			
			\path[->,blue] (11) edge node[left]{$0$}(6);
			\path[->,red] (11) edge[bend left=5] node[above]{$1$} (9);
			
			\path[->,blue] (12) edge[bend left=10] node[below]{$0$}(8);
			\path[->,red] (12) edge node[left]{$1$}(7);
			
			\path[->,blue] (13) edge[bend right=5] node[left]{$0$}(10);
			\path[->,red] (13) edge[bend right=15] node[below]{$1$} (5);
			
			\path[->,blue] (14) edge[bend left=20] node[below]{$0$}(12);
			\path[->,red] (14) edge node[left]{$1$}(3);
			
			\path[->,blue] (15) edge node[below]{$0$} (14);
			\path[->,red] (15) edge[bend left=10] node[left]{$1$}(1);
			
		\end{tikzpicture}
		
		\vspace{0.5cm}
		\begin{tikzpicture}
			[
			> = stealth,
			shorten > = 1pt,
			auto,
			node distance = 1.9cm,
			semithick
			]
			
			\tikzstyle{every state}=
			\node[rectangle,fill=white,draw,rounded corners,minimum size = 4mm]
			
			\node[state,fill=lightgray] (1) {$1$};
			\node[state,fill=lightgray] (2) [right of=1] {$2$};
			\node[state,fill=lightgray] (3) [right of=2] {$3$};
			\node[state,fill=lightgray] (4) [right of=3] {$4$};
			\node[state,fill=lightgray] (5) [right of=4] {$5$};
			
			\node[state,fill=lightgray] (6) [below of=1] {$6$};
			\node[state,fill=lightgray] (7) [right of=6] {$7$};
			\node[state,fill=lightgray] (8) [right of=7] {$8$};
			\node[state,fill=lightgray] (9) [right of=8] {$9$};
			\node[state,fill=lightgray] (10) [right of=9] {$10$};
			
			\node[state,fill=lightgray] (11) [below of=6] {$11$};
			\node[state,fill=lightgray] (12) [right of=11] {$12$};
			\node[state,fill=lightgray] (13) [right of=12] {$13$};
			\node[state,fill=lightgray] (14) [right of=13] {$14$};
			\node[state,fill=lightgray] (15) [right of=14] {$15$};
			
			\path[->,blue] (1) edge node[above]{$0$}(2);
			
			\path[->,blue] (2) edge[bend left=20] node[above]{$0$}(4);
			
			\path[->,red] (3) edge node[above]{$1$} (9);
			
			\path[->,blue] (4) edge node[below]{$0$}(8);
			
			\path[->,blue] (5) edge node[right]{$0$}(10);
			
			\path[->,blue] (6) edge[bend right=10] node[below]{$0$}(12);
			
			\path[->,red] (7) edge node[left]{$1$}(1);
			
			\path[->,red] (8) edge[bend left=5] node[right]{$1$}(15);
			
			\path[->,red] (9) edge node[left]{$1$}(13);
			
			\path[->,red] (10) edge[bend right=3] node[above]{$1$}(11);
			
			\path[->,blue] (11) edge node[left]{$0$} (6);
			
			\path[->,red] (12) edge node[left]{$1$}(7);
			
			\path[->,red] (13) edge[bend right=15] node[below]{$1$}(5);
			
			\path[->,red] (14) edge node[left]{$1$} (3);
			
			\path[->,blue] (15) edge node[below]{$0$}(14);
			
		\end{tikzpicture}
		\caption{{\bf Top}: The graph $\Gamma_4$ with blue arcs based on (\ref{eq:double}) and red arcs based on (\ref{eq:complement}). {\bf Bottom}: A Hamiltonian cycle $\mathcal{H}$ in $\Gamma_4$.}\label{fig:Graph}
	\end{figure}
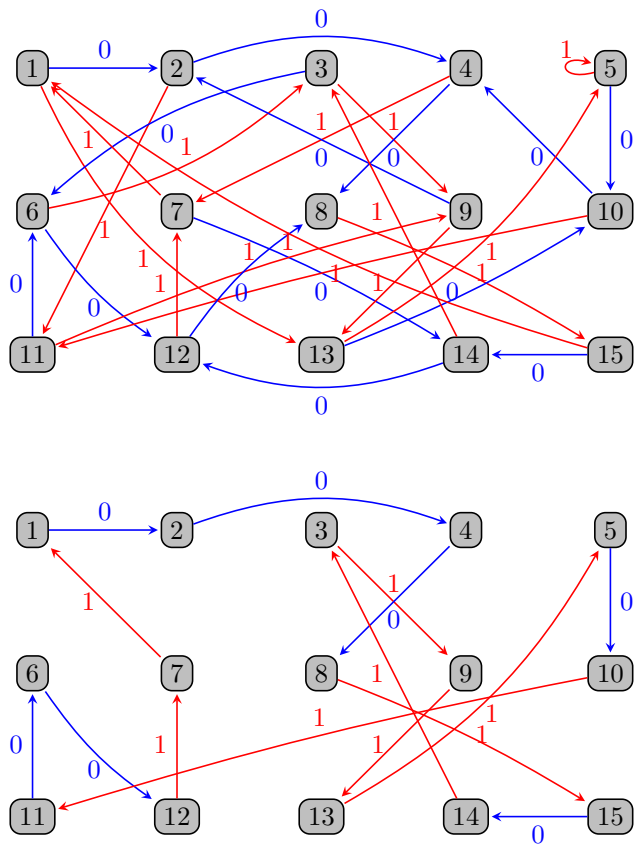
	
	\begin{example}
		The graph $\Gamma_4$ is in Figure~\ref{fig:Graph} {\bf Top}. The loop is from $A=5$ to itself.
	\end{example}
	
	\begin{definition}\label{def:walk}
		Let $\displaystyle{F(x) := \sum_{i=0}^{2^n-2}x^i}$ as in (\ref{eq:main}). The indexed set 
		\[
		W_g := \{(x^i \, g(x) \Mod{F(x)}) \Mod{x^n}~:~i \geq 0\}
		\]
		forms a walk generated by $g(x)$ and we say that $g(x)$ generates the walk $W_g$. 
	\end{definition}
	
	Our task is to identify polynomials $g(x)$, with $g(0) = 1$ and $\gcd(g(x),F(x))=d(x)$, that generate Hamiltonian cycles in $\Gamma_n$. Once such a $g(x)$ is found, Theorem \ref{teotan3} concludes that the reciprocal polynomial $f^{*}(x)$ of $\displaystyle{f(x) := \frac{F(x)}{d(x)}}$ is the minimal polynomial of a modified binary de Bruijn sequence. 
	
	\section{Hamiltonian Cycles in $\Gamma_n$}\label{sec:hamiltonian}
	
	Deciding whether a directed graph is Hamiltonian is hard. A survey on this topic was done by K\"uhn and Osthus in~\cite{Kuhn2012}. Further references and discussion can be found in \cite[Section 6.1]{BangJensen2009}. Fortunately, $\Gamma_n$ has many nice properties that allow us to explicitly determine some Hamiltonian cycles for all $n$. 
	
	\subsection{Hamiltonian Cycles by Two Greedy Algorithms}\label{subsec:greedy}
	
	Inspired by some greedy algorithms in the construction of certain classes of de Bruijn sequences discussed by Chang, Ezerman, and Fahreza in \cite{Chang2020}, we devise two basic algorithms. Algorithm \ref{alg:complement} \emph{prefers the complement over the double arcs} when moving from the current vertex to the next vertex. Algorithm \ref{alg:double} \emph{swaps the preference}, with a modification imposed to avoid the inclusion of $0$, since $0 \notin \Omega$.
	
	\begin{algorithm}[h!]
		\caption{{\tt Prefer Complement}}
		\label{alg:complement}
		\begin{algorithmic}[1]
			\renewcommand{\algorithmicrequire}{\textbf{Input:}}
			\renewcommand{\algorithmicensure}{\textbf{Output:}}
			\Require The order $n$ of the sequences.
			\Ensure Paths in $\Gamma_n$.
			\For{$i \in \{1,2,\ldots,2^n-1\}$} 
			\State{Initiate the indexed set $S=\{i\}$}\Comment{$v_{\rm init}= i$}
			\State{$j \gets$ element in $S$ with largest index}
			\State{$d \gets 2j \Mod{2^n}$}
			\State{$c \gets 2^n-1-d$}
			\While{$c \notin S$ or $d \notin S$}
			\If{$c \notin S$}
			\State{Append $c$ to $S$}
			\Else
			\State{Append $d$ to $S$}
			\EndIf
			\EndWhile
			\State{Path is $(s_1,s_2), \, (s_2,s_3), \, \ldots,\, (s_{\ell-1},s_{\ell})$, where $S:=\{s_1=i,s_2,\ldots,s_{\ell-1},s_{\ell}\}$.}
			\EndFor
		\end{algorithmic}
	\end{algorithm}
	
	\begin{algorithm}[h!]
		\caption{{\tt Modified Prefer Double}}
		\label{alg:double}
		\begin{algorithmic}[1]
			\renewcommand{\algorithmicrequire}{\textbf{Input:}}
			\renewcommand{\algorithmicensure}{\textbf{Output:}}
			\Require The order $n$ of the sequences.
			\Ensure Paths in $\Gamma_n$.
			\For{$i \in \{1,2,\ldots,2^n-1\}$} 
			\State{Initiate indexed set $S=\{i\}$}\Comment{$v_{\rm init}= i$}
			\State{$j \gets$ element in $S$ with largest index}
			\State{$d \gets 2j \Mod{2^n}$}
			\State{$c \gets 2^n-1-d$}
			\While{$c \notin S$ or $d \notin S$}
			\If{$d =0 $ and $c \notin S$}
			\State{Append $c$ to $S$}
			\Else
			\If{$d \notin S$}
			\State{Append $d$ to $S$}
			\Else
			\State{Append $c$ to $S$}
			\EndIf
			\EndIf
			\EndWhile
			\State{Path is $(s_1,s_2), \, (s_2,s_3), \, \ldots,\, (s_{\ell-1},s_{\ell})$, where $S:=\{s_1=i,s_2,\ldots,s_{\ell-1},s_{\ell}\}$.}
			\EndFor
		\end{algorithmic}
	\end{algorithm}
	
	The two algorithms produce paths, starting from an initial vertex $v_{\rm init}$. While each initial vertex produces a path, only several of them lead to Hamiltonian cycles. Table~\ref{table:success} lists the Hamiltonian cycles produced for $n \in \{4,5,6\}$. 
	
	For a given $n \geq 4$, Algorithm \ref{alg:complement} yields $n-1$ distinct, that is, cyclically-inequivalent, Hamiltonian cycles: $n-2$ of them have two distinct initial vertices, whereas $1$ cycle has $3$ possible initial vertices. Let $1 \leq j < n$. Then $v_{\rm init} \in \{2^j-1,2^n - 2^{j-1}\}$ for each $j$. When $v_{\rm init} = 2^{n-1}$ the generated cycle is the same as the one produced when $j=1$, that is, with $v_{\rm init} \in \{1,\,2^{n}-1\}$. This is because the only arc from $2^{n-1}$ is to $2^n-1$ and the complement of the double of $2^n-1$ is $1= (2^{n}-1)-(2^n - 2)$. Algorithm \ref{alg:double} produces a unique Hamiltonian cycle. It occurs if and only if the initial vertex is either $\floor{(2^n-1)/3}$ or its complement $(2^n-1)-\floor{(2^n-1)/3}$. 
	
	\begin{example}
		For $n=4$, Algorithm~\ref{alg:complement} yields the Hamiltonian cycle $\widetilde{\mathcal{H}}$ in Figure~\ref{fig:ham} on $v_{\rm init}=1$. The path consists of arcs
		\begin{align*}
			&(1, 13),\, (13,5),\, (5,10),\, (10,11),\, (11,9),\, (9,2),\, (2,4),\\
			&(4,7),\,(7,14), \,(14,3),\, (3,6),\, (6,12),\, (12,8),\, (8,15).
		\end{align*}
		
		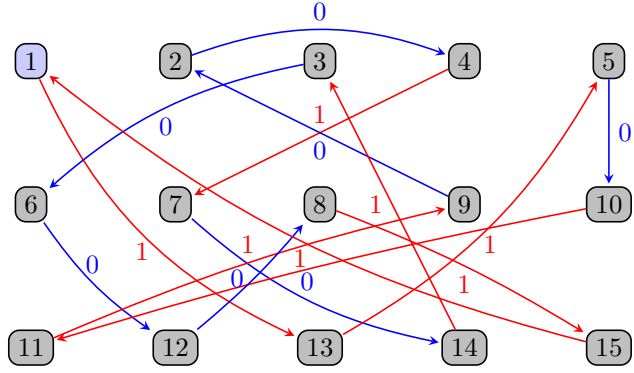
\begin{figure}
			\centering
			\begin{tikzpicture}
				[
				> = stealth,
				shorten > = 1pt,
				auto,
				node distance = 1.9cm,
				semithick
				]
				
				\tikzstyle{every state}=
				\node[rectangle,fill=white,draw,rounded corners,minimum size = 4mm]
				
				\node[state,fill=blue!20] (1) {$1$};
				\node[state,fill=lightgray] (2) [right of=1] {$2$};
				\node[state,fill=lightgray] (3) [right of=2] {$3$};
				\node[state,fill=lightgray] (4) [right of=3] {$4$};
				\node[state,fill=lightgray] (5) [right of=4] {$5$};
				
				\node[state,fill=lightgray] (6) [below of=1] {$6$};
				\node[state,fill=lightgray] (7) [right of=6] {$7$};
				\node[state,fill=lightgray] (8) [right of=7] {$8$};
				\node[state,fill=lightgray] (9) [right of=8] {$9$};
				\node[state,fill=lightgray] (10) [right of=9] {$10$};
				
				\node[state,fill=lightgray] (11) [below of=6] {$11$};
				\node[state,fill=lightgray] (12) [right of=11] {$12$};
				\node[state,fill=lightgray] (13) [right of=12] {$13$};
				\node[state,fill=lightgray] (14) [right of=13] {$14$};
				\node[state,fill=lightgray] (15) [right of=14] {$15$};
				
				\path[->,red] (1) edge[bend right=20] node[below]{$1$} (13);
				\path[->,blue] (2) edge[bend left=20] node[above]{$0$}(4);
				
				\path[->,blue] (3) edge[bend right=15] node[below]{$0$}(6);
				\path[->,red] (4) edge node[above]{$1$} (7);
				
				\path[->,blue] (5) edge node[right]{$0$}(10);
				
				\path[->,blue] (6) edge[bend right=10] node[above]{$0$} (12);
				
				\path[->,blue] (7) edge[bend right=15] node[above]{$0$} (14);
				
				\path[->,red] (8) edge[bend left=5] node[below]{$1$} (15);
				
				\path[->,blue] (9) edge node[below]{$0$} (2);
				
				\path[->,red] (10) edge[bend right=3] (11);
				
				\path[->,red] (11) edge[bend left=5] node[above]{$1$} (9);
				
				\path[->,blue] (12) edge[bend right=5] node[left]{$0$} (8);
				
				\path[->,red] (13) edge[bend right=15] node[below]{$1$} (5);
				
				\path[->,red] (14) edge node[left]{$1$} (3);
				
				\path[->,red] (15) edge[bend left=12] node[below]{$1$} (1);
				
			\end{tikzpicture} 
			\caption{The Hamiltonian cycle $\widehat{\mathcal{H}}$ in $\Gamma_4$ produced by Algorithm \ref{alg:complement} on $v_{\rm init}=1$. It is generated by $g(x)=x^{10}+x^8 + x^5 + x + 1$. The resulting sequence $(0, 0, 0, 1,  1, 1, 1, 0, 1, 1, 0, 0, 1, 0,1)$ has linear complexity $14$, which is maximal.}\label{fig:ham}
		\end{figure}
	\end{example}
	
	\begin{table*}[t!]
		\caption{Hamiltonian Cycles constructed by Algorithms \ref{alg:complement} and \ref{alg:double} for $n \in \{4,5,6\}$.}
		\label{table:success}
		\centering
		\begin{tabular}{ccl}
			\toprule
			$n$ & $v_{\rm init}$ & \multicolumn{1}{c}{The resulting Hamiltonian cycle} \\ 
		\midrule
		\multicolumn{3}{c}{Algorithm \ref{alg:complement} : {\tt Prefer Complement}} \\
		\midrule
		$4$ & $1,15,8$ & $(1, 13, 5, 10, 11, 9, 2, 4, 7, 14, 3, 6, 12, 8, 15)$ \\
		& $3,14$ & $(3, 9, 13, 5, 10, 11, 6, 12, 7, 1, 2, 4, 8, 15, 14)$ \\
		& $7,12$ & $(7, 1, 13, 5, 10, 11, 9, 2, 4, 8, 15, 14, 3, 6, 12)$ \\
		
		$5$ & $1,31,16$ & $(1, 29, 5, 21, 10, 11, 9, 13, 26, 20, 23, 17, 2, 27, 22,$\\
		&& $~19, 25, 18, 4, 8, 15,30, 3, 6, 12, 7, 14, 28, 24, 16, 31)$ \\
		& $3,30$ & $(3, 25, 13, 5, 21, 10, 11, 9, 18, 27, 22, 19, 6, 12, 7,$\\
		&& $~17, 29, 26, 20, 23,
		14, 28, 24, 15, 1, 2, 4, 8, 16, 31, 30)$ \\
		& $7,28$ & $(7, 17, 29, 5, 21, 10, 11, 9, 13, 26, 20, 23, 14, 3, 25,$\\
		&& $~18, 27, 22, 19, 6,
		12, 24, 15, 1, 2, 4, 8, 16, 31, 30, 28)$ \\
		& $15,24$ & $(15, 1, 29, 5, 21, 10, 11, 9, 13, 26, 20, 23, 17, 2, 27,$\\
		&& $~22, 19, 25, 18, 4, 8,16, 31, 30, 3, 6, 12, 7, 14, 28, 24)$ \\
		
		$6$ & $1,63,32$ & $(1, 61, 5, 53, 21, 42, 43, 41, 45, 37, 10, 20, 23, 17, 29, 58, 11, 22, 19, 25,$\\
		&&$~13, 26, 52, 40, 47, 33, 2, 59, 9, 18, 27, 54, 44, 39, 49, 34,4, 55, 46, 35, 57,$\\
		&&$~50, 36, 8, 16, 31, 62, 3, 6, 51, 38, 12, 24, 15, 30, 60, 7, 14, 28, 56, 48, 32, 63)$ \\
		
		& $3,62$ & $(3, 57, 13, 37, 53, 21, 42, 43, 41, 45, 26, 11, 22, 19, 25, 50, 
		27, 9, 18, 36,$\\
		&& $~55, 17, 29, 5, 10, 20, 23, 46, 35, 6, 51, 38,
		12, 39, 49, 34, 59, 54, 44, 24, 15, $\\
		&&$~33, 61, 58, 52, 40, 47, 30, 60, 7, 14, 28, 56, 48, 31, 1, 2, 4, 8, 16, 32, 63, 62)$ \\
		
		& $7,60$ & $(7, 49, 29, 5, 53, 21, 42, 43, 41, 45, 37, 10, 20, 23, 17, 34, 59, 9, 18, 27, 54,$\\
		&& $~19, 25, 13, 26, 11, 22, 44, 39, 14, 35, 57, 50, 36, 55, 46, 28, 56, 15, 33,61, $\\
		&& $~58, 52, 40, 47, 30, 3, 6, 51, 38, 12, 24, 48, 31, 1, 2, 4, 8, 16, 32, 63, 62, 60)$\\
		
		& $15,56$ & $(15, 33, 61, 5, 53, 21, 42, 43, 41, 45, 37, 10, 20, 23, 17, 29, 58, 11, 22,
		19,$\\
		&&$~25, 13, 26, 52, 40, 47, 30, 3, 57, 50, 27, 9, 18,36, 55, 46, 35, 6, 51, 38,12,$\\
		&&$~39, 49, 34, 59, 54, 44, 24, 48, 31, 1, 2, 4, 8, 16, 32, 63, 62, 60, 7, 14, 28, 56)$ \\
		
		& $31,48$ & $(31, 1, 61, 5, 53, 21, 42, 43, 41, 45, 37, 10, 20, 23, 17, 29, 58, 11, 22, 19,25,$\\
		&&$~13, 26, 52, 40, 47, 33, 2, 59, 9, 18, 27,54, 44, 39, 49,34, 4, 55, 46, 35, 57$\\
		&&$~50, 36, 8, 16, 32, 63, 62, 3, 6, 51, 38, 12, 24, 15, 30, 60, 7, 14, 28, 56, 48)$ \\
		
		\bottomrule		
		\multicolumn{3}{c}{Algorithm \ref{alg:complement} : {\tt Modified Prefer Double}} \\
		\midrule
		$4$ & $5,10$ & $(5, 10, 4, 8, 15, 14, 12, 7, 1, 2, 11, 6, 3, 9, 13)$ \\
		$5$ & $10,21$ & $(10, 20, 8, 16, 31, 30, 28, 24, 15, 1, 2, 4, 23, 14, 3, $\\
		&& $~6,
		12, 7, 17, 29, 26, 11, 22, 19, 25, 18, 27, 9, 13, 5, 21)$ \\
		$6$ & $21,42$ & $(21, 42, 20, 40, 16, 32, 63, 62, 60, 56, 48, 31, 1, 2, 4, 8,47,30,3, 6, 12, 24, $\\
		&& $~15,33, 61, 58, 52, 23,46, 28, 7, 14, 35, 57, 50, 36, 55, 17, 34, 59, 54,44, 39,$\\
		&& $~49, 29, 5, 10, 43, 22, 19, 38, 51, 25,
		13, 26, 11, 41, 18, 27, 9, 45, 37, 53)$ \\		
		\bottomrule
	\end{tabular}
\end{table*}

For brevity we will only prove, in Theorem \ref{thm:fullcycle}, that Algorithm \ref{alg:complement} on $v_{\rm init}=1$ always yields a Hamiltonian cycle. The respective proofs for the other valid initial vertices and on Algorithm \ref{alg:double} when using the two specified initial vertices follow a similar line of reasoning. To identify Hamiltonian cycles beyond those produced by the two algorithms, we establish a general result on the paths produced by Algorithm \ref{alg:complement}. The result will be used in the next subsection to identify many more Hamiltonian cycles in $\Gamma_n$ by cycle joining.

\begin{theorem}\label{prefercom}
	Let $V(\Gamma_n) := \{1,\ldots,2^n-1\}$. Given an indexed set whose elements are vertices in $\Gamma_n$ in the form of 
	\begin{equation}\label{eq:Theta}
		\Theta := \{\alpha_1, \, \ldots, \, \alpha_{2^n-1} \; : \; \alpha_i \in V \mbox{ and } \alpha_i \ne \alpha_j \mbox{ for } i \ne j\}, 
	\end{equation}
	let $\Psi$ be a mapping on $\Theta$ defined by, for $1 \leq i < 2^n$,
	\begin{equation}\label{eq:Psi}
		\Psi: \alpha_i \mapsto 
		\begin{cases} 
			\beta_i:=(2^n-1)- (2 \alpha_i \Mod{2^n}) 
			\mbox{, if } \alpha_i \mbox{ had not been mapped to } \beta_i,\\
			2 \alpha_i \Mod{2^n} \mbox{, otherwise.} 
		\end{cases}
	\end{equation}
	Then $\Psi$ is a permutation on $\Theta$.
\end{theorem}

\begin{proof}
	We show that $\Psi$ is a bijection on $\Theta$. It is immediate to verify that there exist indices $1 \leq k \neq \ell < 2^n$ such that $\alpha_{\ell}=(\alpha_k + 2^{n-1}) \Mod{2^n}$. If $k < \ell$, then 
	\begin{equation}
		\Psi(\alpha_k)= (2^n-1) - (2 \alpha_k \Mod{2^n}) 
		\neq \Psi(\alpha_{\ell}) = 2 \alpha_k \Mod{2^n}.
	\end{equation}
	If $\ell < k$, then 
	\begin{equation}
		\Psi(\alpha_{\ell})= (2^n-1) - (2 \alpha_k \Mod{2^n}) 
		\neq \Psi(\alpha_k) = 2 \alpha_k \Mod{2^n}.
	\end{equation}
	Thus, $\Psi$ is injective. 
	
	Let $\alpha$ be an arbitrarily selected element of $\Theta$. If $\alpha$ is even, then there exists an integer $k \in \left\{1, 2,\ldots, \frac{2^n-2}{2} \right\}$ such that $\alpha =2k$. Hence, either $\Psi(k) = \alpha$ or $\Psi(k + 2^{n-1}) = \alpha$. If $\alpha$ is odd, then there exists an integer 
	$t \in \left\{0,1, 2,\ldots, \frac{2^n-2}{2} \right\}$ such that $\alpha = 2t + 1$. Then, either $\Psi\left(\frac{2^n-2t-2}{2}\right)=\alpha$ or $\Psi\left(\frac{2^n-2t-2}{2} + 2^{n-1}\right)= \alpha$. We conclude that $\Psi$ is surjective and the proof is now complete.
\end{proof}
Since the function $\Psi$ in (\ref{eq:Psi}) is a permutation on a finite set $\Theta$, then $\Psi$ can be written as a composition of $j$ disjoint cycles 
\begin{equation}\label{eq:disj}
	\Psi = \mathcal{C}_1 \circ \mathcal{C}_2 \circ \cdots \circ \mathcal{C}_j. 
\end{equation}

\begin{definition}\label{def:cycle}
	Let $\left(\alpha_{k,1},\alpha_{k,2},\ldots,\alpha_{k,\ell}\right)$ be a cycle $\mathcal{C}_k$ generated by $\Psi$. We say that $\mathcal{C}_k$ \emph{starts} at $\alpha_{k,1}$ and \emph{ends} at $\alpha_{k,\ell}$ since, by then, both possible images $\Psi(\alpha_{k,\ell})$, one of which being $\alpha_{k,1}$, have all appeared. The elements $\alpha_{k,1}$ and $\alpha_{k,\ell}$ are, respectively, the {\it starting element} and the {\it terminating element} of $\mathcal{C}_k$.
\end{definition}

\begin{lemma}\label{lemma:ordering}
	Let the indexed set $\Theta$ and the function $\Psi$ be as defined in Theorem \ref{prefercom}. Let $\gamma \in \Theta$ be an even number which is not a starting element in any cycle. Then the complement $2^n-1- \gamma$ of $\gamma$ is an odd number that occurs before $\gamma$ in the said cycle.
\end{lemma}
\begin{proof}
	By how $\Psi$ is defined and since $\gamma \in \Theta$ is not the starting element, it is impossible for $\gamma$ to appear in the cycle before its complement, which is an odd number. 
\end{proof}

\begin{lemma}\label{lemma:halving}
	Let $\beta \in \{1,2, \ldots,2^{n-1}-1\}$, then 
	$\{\Psi(\beta),\Psi(\beta + 2^{n-1})\} = \{2^n-1 - 2 \beta, 2\beta\}$. By the time both $2^n-1 - 2 \beta$ and $2\beta$ appear in a cycle $\mathcal{C}$, we know that both $\beta$ and $\beta + 2^{n-1}$ must have appeared earlier.
\end{lemma}
\begin{proof}
	Since $\Psi$ is bijective, the appearance of both $2^n-1 - 2 \beta$ and $2\beta$ requires prior inclusion of both possible predecessors in the cycle.
\end{proof}

The conclusion that $\Gamma_n$ is Hamiltonian for all $n \geq 4$ follows from the next theorem.
\begin{theorem}\label{thm:fullcycle}
	Let the set $\Theta$ and function $\Psi$ be as defined in Theorem \ref{prefercom}. The function $\Psi$ produces a single cycle of length $2^n-1$ that starts at $1$ and ends at $2^n-1$.
\end{theorem}

\begin{proof}
	Let $\mathcal{C}$ be the circle 
	\begin{equation}\label{equ:basic}
		(1,\Psi(1), \Psi(\Psi(1)),\ldots,\delta,\alpha).
	\end{equation}
	Consider the shift-equivalent cycle 
	$(\Psi(\delta)=\alpha, \Psi(\alpha)=1, \Psi(1), \ldots, \delta)$. Since $\Psi(\alpha)=1$, $\alpha$ must be either $2^{n-1}-1$ or $2^n-1$. 
	
	For a contradiction, let us assume that $\alpha=2^{n-1}-1$. Since $1$ has already appeared in $\mathcal{C}$, the only other possible successor of $2^{n-1}-1$, namely $2^n-2$, must have appeared in $\mathcal{C}$. The two possible predecessors of $2^n-2$ are $2^{n-1}-1$ and $2^n-1$. Hence, $\mathcal{C}$ must have the form
	\begin{equation}\label{eq:suppose}
		\mathcal{C} = (1, \ldots, 2^n-1, 2^n-2, \ldots, 2^{n-1}-1).
	\end{equation}
	On the other hand, since  $2^{n-1}-1$ is both an odd number and the terminating element, its complement, namely $2^{n-1}$, must not have appeared in $\mathcal{C}$. But this rules out $2^n-1$ from $\mathcal{C}$ as well, contradicting (\ref{eq:suppose}). Thus, the terminating element must be $\alpha=2^n-1$.
	
	We now show that all elements of $\Theta$ appear in $\mathcal{C}$. The computations are done modulo $2^n$. It is clear that $\delta = 2^{n-1}$, since it is the only preimage of $2^n-1$. Hence, we have
	\begin{equation}\label{eq:Correct}
		\mathcal{C} = (1, \, \Psi(1),\, \ldots, \, 2^{n-1}, \, 2^n-1).
	\end{equation}
	
	Aided by Lemmas \ref{lemma:ordering} and \ref{lemma:halving}, we proceed by induction to confirm that each even number $2 \leq k \leq 2^n-2$ appears.
	\begin{itemize}
		\item  Since $\delta=2^{n-1}$ is an even number, Lemma~\ref{lemma:ordering} says that its complement $\delta^*:=2^n -2^{n-1}-1 = 2^{n-1}-1$ appears. By Lemma~\ref{lemma:halving}, both $2^{n-2}$ and $2^{n-1}+2^{n-2}$ must have appeared before both $\delta$ and $\delta^{*}$ do.
		\item Since $2^{n-2}$ is even, by Lemma~\ref{lemma:ordering}, its complement appear. We infer by Lemma~\ref{lemma:halving} that $2^{n-3}$ and $2^{n-1}+2^{n-3}$ appear. 
		\item We repeat the same reasoning on $2^{n-1}+2^{n-2}$. Its complement appears and so do both $2^{n-2}+2^{n-3}$ and $2^{n-1}+ 2^{n-2}+2^{n-3}$.
		\item Continuing the process, we establish the appearance of
		\begin{align}
			& 2, \, 2^2, \, \ldots, \, 2^{n-1}, \notag \\
			& 2 + 2^{n-1}, \, 2^2 + 2^{n-1}, \, \ldots, \, 2^{n-2} + 2^{n-1}, \notag \\
			& 2 + 2^{n-2}+ 2^{n-1}, \, 2^2 + 2^{n-2}+ 2^{n-1}, \, \ldots,  2^{n-3} + 2^{n-2} + 2^{n-1}, \notag \\
			& \ldots, \, \ldots , \, \ldots, \, \ldots, \notag \\
			&2 + 2^2 + \ldots + 2^{n-2} + 2^{n-1},
		\end{align}
		which cover all even numbers in the desired range.
	\end{itemize}
	Lemma~\ref{lemma:ordering} ensures the appearance of each odd number $j$ such that $3 \leq j \leq (2^n-3)$, completing the proof.
\end{proof}

\begin{corollary}\label{cor:Hamilton}
	The graph $\Gamma_n$ for each $n \geq 4$ is Hamiltonian.
\end{corollary}
\begin{proof}
	We have built $\mathcal{C}=(1,\Psi(1),\Psi(\Psi(1)),\ldots, 2^{n-1},2^n-1)$ in the proof of Theorem~\ref{thm:fullcycle}. Following the sequence of vertices in $\mathcal{C}$, a Hamiltonian path in $\Gamma_n$ is formed by the edges
	\[
	(1,\Psi(1)),\, (\Psi(1),\Psi(\Psi(1))), \, \ldots, \, (2^{n-1},2^n-1).
	\]
\end{proof}

\subsection{More Hamiltonian Cycles by Cycle Joining}\label{subsec:CJM}

Theorem \ref{thm:fullcycle} guarantees that, starting from $v_{\rm init}=1$, Algorithm \ref{alg:complement} produces a cycle of length $2^n-1$. However, for most other $v_{\rm init}$, we obtain disjoint cycles as in (\ref{eq:disj}) with $j > 1$. The following lemma gives a condition for when two disjoint cycles can be joined into a longer cycle.

\begin{lemma}\label{lemma:cjm}
	Two disjoint cycles $\mathcal{C}$ and $\widehat{\mathcal{C}}$ can be joined into one cycle if there exist $c \in \mathcal{C}$ and $d \in \widehat{\mathcal{C}}$ such that $d$ is the complement of $c$, that is $d = 2^n-1 - c$.
\end{lemma}
\begin{proof}
	Let $\mathcal{C}=(c_1,\, \ldots,\, c_{i-1},\, c_{i}, \, c_{i+1},\, \ldots, c_k)$ and $\widehat{\mathcal{C}}=(d_1,\, \ldots,\, d_{j-1},\, d_{j}, \, d_{j+1},\, \ldots, d_{\ell})$, with $c_i$ and $d_j$ forming a complementary pair, that is, $c_i + d_j = 2^n-1$. We exchange the predecessors of $c_i$ and $d_j$ to obtain the joined cycle
	\[
	(c_1,\, \ldots,\, c_{i-1},\, d_j, \, d_{j+1}, \, \ldots, d_{\ell}, d_1, \, \ldots, d_{j-1}, \, c_{i}, \, c_{i+1},\, \ldots, c_k).
	\]
\end{proof}

\begin{theorem}\label{thm:cyclejoin}
	If $\Psi$ generates disjoint cycles as in (\ref{eq:disj}), then all of the cycles can be joined into a single cycle of length $2^n-1$.
\end{theorem}
\begin{proof}
	If $\Psi$ generates only one cycle, then it is clear that the length of the cycle must be $2^n-1$. Suppose that $\Psi$ generates at least two cycles and we take any cycle $\mathcal{C}$. It suffices to show that there exists $c \in \mathcal{C}$ whose complement $\widehat{c}:=2^n-1-c$ does not appear in $\mathcal{C}$. For a contradiction, let there be no such element. Hence, every element in $\mathcal{C}$ has its complement in $\mathcal{C}$, that is, $\widehat{c} \in \mathcal{C}$ for any $c \in \mathcal{C}$. By definition, exactly one of either $c$ or $\widehat{c}$ is an even number $e$. Therefore, the predecessors of $c$ and $\widehat{c}$, namely, $k := \frac{e}{2}$ and $2^{n-1} + k$ also appear in $\mathcal{C}$. Following this fact, all integers in $\{1,2,..,2^n-1\}$ appear in $\mathcal{C}$. This contradicts the assumption that $\Psi$ generates two or more cycles. 
\end{proof}

We can implement Theorem \ref{thm:cyclejoin} and enumerate the number of resulting Hamiltonian cycles that can be constructed by adopting the {\it cycle joining method} from the theory of feedback shift registers. 

Let $\Psi$ be expressed in terms of its disjoint cycles as in (\ref{eq:disj}). For distinct $1 \leq i \neq k  \leq j$, let $\tau_{i,k}:=(r,s)$ denote $r \in C_i$ and $s \in C_k$ with $r+s = 2^n-1$. The tuple joins $C_i$ and $C_k$ by interchanging the respective {\it predecessors} of $r$ and $s$. 

To count the number of inequivalent Hamiltonian cycles that can be produced from $\Psi$, we build the associated undirected multigraph $G_{\Psi}$ as follows. The vertices are $\mathcal{C}_1, \ldots, \mathcal{C}_j$. We add an edge labelled $(r,s)$ between $C_{i}$ and $C_k$ whenever there is a pair $(r,s)$ with the property that $r \in C_i$, $s \in C_k$, and $r+s=2^n-1$. The graph $G_{\Psi}$ has no loops but may have multiple edges connecting the same pair of vertices. The number of Hamiltonian cycles that can be constructed in this manner is equal to the number of subgraphs of $G_{\Psi}$ which are {\it rooted spanning trees}.    

The following well-known counting formula is a variant of the BEST (de {\bf B}ruijn, {\bf E}hrenfest, {\bf S}mith, and {\bf T}utte) Theorem. More detail on graphical approaches to the generation of full cycles, including the BEST Theorem and its history, can be found in \cite[Section~2]{Fred82}. The {\it cofactor} of entry $m_{i,k}$ in a matrix $M:=(m_{i,k})$ is $(-1)^{i+k}$ times the determinant of the matrix obtained by deleting the $i^{\rm th}$ row and $k^{\rm th}$ column of $M$.

\begin{theorem}(BEST)\label{BEST} Let $V:=\{\mathcal{C}_1,\ldots, \mathcal{C}_j\}$ be the vertex set of $G_{\Psi}$. Let $M=(m_{i,k})$ be the $j \times j$ matrix derived from $G_{\Psi}$ in which $m_{i,i}$ is the number of edges incident to $\mathcal{C}_i$ and $m_{i,k}$ is the negative of the number of edges between vertices $\mathcal{C}_i$ and $\mathcal{C}_k$ for $i \neq k$. Then the number of rooted spanning trees of $G_{\Psi}$ is the cofactor of any entry of $M$.
\end{theorem}

\begin{example}\label{ex:GPsi}
	A randomized instance for $n=4$ picks
	\[
	\Psi = (6, 3, 9, 13, 5, 10, 11) \circ (4, 7, 1, 2) \circ (14, 12, 8, 15).
	\]
	We label the cycles from left to right as $\mathcal{C}_1$, $\mathcal{C}_2$, and $\mathcal{C}_3$ to get the associated graph $G_{\Psi}$ in Figure \ref{fig:GPsi}.
	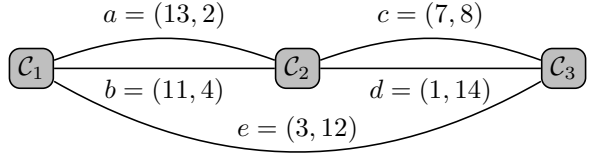
\begin{figure}
		\centering
		\begin{tikzpicture}
			[
			auto,
			node distance = 3.5cm,
			semithick
			]
			
			\tikzstyle{every state}=
			\node[rectangle,fill=lightgray,draw,rounded corners,minimum size = 5mm]
			
			\node[state] (1) {$\mathcal{C}_1$};
			\node[state] (2) [right of=1] {$\mathcal{C}_2$};
			\node[state] (3) [right of=2] {$\mathcal{C}_3$};
			
			\path[-] (1) edge node[below]{$b=(11,4)$}(2);
			\path[-] (1) edge[bend left=20] node[above]{$a=(13,2)$} (2);
			
			\path[-] (2) edge node[below]{$d=(1,14)$}(3);
			\path[-] (2) edge[bend left=20] node[above]{$c=(7,8)$} (3);
			\path[-] (1) edge[bend right=30] node[above]{$e=(3,12)$}(3);
			
		\end{tikzpicture}
		\caption{The associated graph $G_{\Psi}$ for Example \ref{ex:GPsi}.}\label{fig:GPsi}
	\end{figure}
	
	The associated matrix is
	\[
	M:=
	\begin{bmatrix*}[r]
		3 & -2 & -1\\
		-2 & 4 & -2\\
		-1 & -2 & 3
	\end{bmatrix*}.
	\]
	By BEST Theorem, there are $8$ Hamiltonian cycles that can be constructed from $\Psi$. Table~\ref{tab:8H} list them. All but one of the sequences have maximal linear complexity $14$. 
	\begin{table*}[h!]
		\caption{The $8$ Hamiltonian cycles produced from $\Psi$ by cycle joining in Example \ref{ex:GPsi}. For $n=4$ we have $F(x)= 1 + x + x^2 + \ldots + x^{14}$.}\label{tab:8H}
		\centering
		\resizebox{\textwidth}{!}{%
			\begin{tabular}{cccl}
				\toprule
				Joining & Hamiltonian Cycle & Modified de Bruijn Sequence $\mathbf{s}$ & $m_{\mathbf{s}}(x)$ \\
				\midrule
				$a,c$ & $(6,3,9,2,4,8,15,14,12,7,1,13,5,10,11)$ & $(1,1,0,0,0,1,0,0,1,1,1,1,0,1,0)$ & $F(x)$ \\
				
				$a,d$ & $(6,12,8,15,14,3,9,2,4,7,1,13,5,10,11)$ & $(0,0,1,0,1,1,0,0,1,1,1,1,0,1,0)$ & $F(x)$ \\
				
				$b,c$ & $(6,3,9,13,5,10,4,8,15,14,12,7,1,2,11)$ & $(1,1,1,1,0,0,0,1,0,0,1,1,0,1,0)$ & $x^4+x+1$ \\
				
				$b,d$ & $(6,3,9,13,5,10,4,7,14,12,8,15,1,2,11)$ & $(1,1,1,1,0,0,1,0,0,0,1,1,0,1,0)$ & $F(x)$ \\
				
				$e,a$ & $(6,12,8,15,14,3,9,2,4,7,1,13,5,10,11)$ & $(0,0,1,0,1,1,0,0,1,1,1,1,0,1,0)$ & $F(x)$ \\
				
				$e,b$ & $(6,12,8,15,14,3,9,13,5,10,4,7,1,2,11)$ & $(0,0,1,0,1,1,1,1,0,0,1,1,0,1,0)$ & $F(x)$ \\
				
				$e,c$ & $(6,12,7,1,2,4,8,15,14,3,9,13,5,10,11)$ & $(0,1,1,0,0,0,1,0,1,1,1,1,0,1,0)$ & $F(x)$ \\
				
				$e,d$ & $(6,12,8,15,1,2,4,7,14,3,9,13,5,10,11)$ & $(0,0,1,1,0,0,1,0,1,1,1,1,0,1,0)$ & $F(x)$ \\
				
				\bottomrule		
			\end{tabular}
		}
	\end{table*}
\end{example}

\section{The Canonical Generator Polynomial}

In this section we show that there exists a canonical generator $c_{\mathcal{H}}(x) \in \mathbb{F}_2[x]$ for every Hamiltonian cycle $\mathcal{H} \in \Gamma_n$. 

The {\it de Bruijn graph}, denoted by $B_n$ or simply $B$ when $n$ is understood, is also known as the {\it Good graph} and {\it de Bruin-Good graph}. It was introduced independently by de Bruijn in \cite{Bruijn46} and by Good in \cite{Good1946}. Its set of vertices consists of binary $n$-strings 
\[
\{v_1,v_2,\ldots, v_n : v_j \in \mathbb{F}_2 \mbox{ for all } 1 \leq j \leq n \}. 
\]
An arc from vertex $c_1, c_2, \ldots ,c_n$ to vertex $c_2, c_3, \ldots, c_{n+1}$ is labelled $0$ and $1$, respectively, if $c_{n+1} = 0$ and $c_{n+1} = 1$. 

\begin{theorem}\label{thm:correspondence}
	If $\mathcal{H}$ is a Hamiltonian cycle in $\Gamma_n$, then $\mathcal{H}$ corresponds to a modified binary de Bruijn sequence.
\end{theorem}

\begin{proof}
	We revert back to the binary string representation of the vertices in $\Gamma_n$. Let $e_1, e_2,\ldots,e_{2^n-1}$ be the labels, each is either $0$ or $1$, on the ordered arcs in $\mathcal{H}$. If $e_i=0$, then the $i^{\rm th}$ arc connects the exact same pair of vertices in both $\Gamma_n$ and in the original de Bruijn graph $\mathcal{B}_n$. Moreover, if $e_j=1$, then it corresponds to an arc with label $1$ in $\mathcal{B}_n$. Thus, the sequence $(e_1,e_2,\ldots,e_{2^n-1})$ that corresponds to $\mathcal{H}$ is a modified de Bruijn sequence.
\end{proof}

\begin{corollary}\label{cor:existence of g(x)}
	If $\mathcal{H}$ is a Hamiltonian cycle in $\Gamma_n$, then there exists a polynomial $g(x) \in \mathbb{F}_2[x]$ that generates $\mathcal{H}$. The consecutive elements in 
	\[
	W_g =\{(x^i~g(x) \Mod{F(x)}) \Mod{x^n} : 0 \leq i < 2^n-1\}
	\]
	forms the Hamiltonian path in $\mathcal{H}$.
\end{corollary}
\begin{proof}
	Since $\mathcal{H}$ is a Hamiltonian cycle in $\Gamma_n$, it corresponds to a modified binary de Bruijn sequence $\mathbf{s}$ which can be represented as $\displaystyle{\frac{g(x)}{F(x)}}$, with $\deg(g(x)) < 2^n-2$. Lemma \ref{teotan2} ensures the existence of the required $g(x)$ and, by Definition \ref{def:walk}, $g(x)$ generates $\mathcal{H}$.
\end{proof}

We recall Lemma \ref{teotan1} before establishing our next result. If $\mathbf{s}$ is a modified binary de Bruijn sequence whose rational fraction representation is $\displaystyle{\frac{g(x)}{F(x)}}$, for some $ g(x) \in \mathbb{F}_2[x]$, then, for any $ 0 \leq k < 2^n-1$, the rational fraction representation of the $k$-shifted sequence $L^k \, \mathbf{s}$ in Equation (\ref{eq:rational}) is given by $\displaystyle{
	\frac{g_k(x)}{F(x)}}$ where 
$g_k(x) := g(x) \, x^{2^n-1-k} \Mod{F(x)}$. Two polynomials $g(x)$ and $h(x)$ may generate the same Hamiltonian cycle $\mathcal{H}$ in $\Gamma_n$. The following lemma explains how to relate the polynomials to one another.

\begin{lemma}\label{polagx}
	Let $n \geq 4$ be given and let 
	$F(x) := \sum_{i=0}^{2^n-2} x^i = 1 + x + x^2+ \ldots + x^{2^n-2}$. Two polynomials $g(x)$ and $h(x)$ in $\mathbb{F}_2[x]$ generate the same Hamiltonian cycle $\mathcal{H}$ in $\Gamma_n$ if and only if 
	\begin{equation}\label{eq:samecycle}
		h(x) = x^k \, g(x) \Mod{F(x)} \mbox{ for some } 0 \leq k < 2^n-1.
	\end{equation}
\end{lemma}

\begin{proof}
	To verify that $x^{2^n-1} \, g(x) \Mod{F(x)} = g(x)$ it suffices to use the fact that $F(x)$ divides $x^{2^n-1}+1$, which implies $x^{2^n-1} \Mod{F(x)} = 1 $.
	
	Let $a(x)$ be any element of $\Omega$. If $g(x)$ and $h(x)$ generate the same Hamiltonian cycle $\mathcal{H}$, then there exist $i$ and $j$, where $0 \leq i,j < 2^n-1$, such that 
	\begin{equation}\label{eq:firstdir}
		a(x) = (x^i \, g(x) \Mod{F(x)}) \Mod{x^n} 
		=(x^j \, h(x) \Mod{F(x)}) \Mod{x^n}.
	\end{equation} 
	If $i \geq j$, then letting $k:=i-j$ yields
	\begin{equation}\label{eq:sameH}
		\left(x^{k} \, g(x) \Mod{F(x)}\right) \Mod{x^n} = h(x).
	\end{equation}
	If $i < j$, then multiplying both sides of (\ref{eq:sameH}) by $x^{2^n-1}$ gives us
	\begin{equation}
		\left(x^{2^n-1+k} \, g(x) \Mod{F(x)}\right) \Mod{x^n} = 
		x^{2^n-1} \, h(x).
	\end{equation}
	Since $x^{2^n-1} \, h(x) \Mod{F(x)} = h(x)$, we conclude that
	\[
	h(x) = \left(x^k \, g(x) \Mod{F(x)}\right) \Mod{x^n}.
	\]
	
	Conversely, let $h(x) = x^k \, g(x) \Mod{F(x)}$ for some $k$ with $0 \leq k < 2^n-1$. We take two vertices $A$ and $B$ in $\Gamma_n$ that correspond respectively to 
	\begin{align*}
		a(x) &= \left(x^i \, h(x) \Mod{F(x)}\right) \Mod{x^n-1} \mbox{ and}\\ 
		b(x) &= \left(x^{i+1} \, h(x) \Mod{F(x)}\right) \Mod{x^n-1}
	\end{align*}
	for some $0 \leq i < 2^n-1$.
	Then
	\begin{align}
		a(x) &= \left(x^i \, x^k \, g(x) \Mod{F(x)}\right) \Mod{x^n} \mbox{ and} \notag\\
		b(x) &= \left(x^{i+1} \, x^k \, g(x) \Mod{F(x)}\right) \Mod{x^n}.
	\end{align}
	Using $j:=(i+k) \Mod{2^n-2}$, one writes 
	\begin{align}
		a(x) &= \left(x^j \, g(x) \Mod{F(x)}\right) \Mod{x^n} \mbox{ and} \notag\\
		b(x) &= \left(x^{j+1} \, g(x) \Mod{F(x)}\right) \Mod{x^n}.
	\end{align}
	Thus, both $g(x)$ and $h(x)$ generate the same $\mathcal{H}$ in $\Gamma_n$.
\end{proof}

\begin{example}\label{ex:genpoly}
	Each of the polynomials 
	\begin{align*}
		g(x) &= x^{10} + x^9 + x^7 + x^5 + x^4 + 1,\\
		x \, g(x) &= x^{11} + x^{10} + x^8 + x^6 + x^5 + x, \\
		x^2 \, g(x) &= x^{12} + x^{11} + x^9 + x^7 + x^6 + x^2, \\
		x^3 \, g(x) &= x^{13} + x^{12} + x^{10} + x^8 + x^7 + x^3, \\
		x^4 \, g(x) &= x^{12} + x^{10} + x^7 + x^6 + x^5 + x^3 + x^2 + x + 1.
	\end{align*}
	generates the Hamiltonian cycle $\mathcal{H}$ in Figure \ref{fig:Graph} {\bf Bottom}.
\end{example}

Lemma \ref{polagx} tells us to focus on the polynomial with the least degree that generates $\mathcal{H}$. We formally define this polynomial.

\begin{definition}\label{def:canonical_generator}
	Let $\mathcal{H}$ be a Hamiltonian cycle in $\Gamma_n$. The polynomial with the least degree $\ell$ among all of the polynomials that generate $\mathcal{H}$ is the {\it canonical generator} $c_{\mathcal{H}}(x)$ of $\mathcal{H}$.
\end{definition}

\begin{theorem}\label{thm:canonical_g(x)}
	Let $n \geq 4$ be given and let $\displaystyle{\frac{g(x)}{F(x)}}$ be a rational fraction representation of a modified binary de Bruijn sequence $\mathbf{s}$. Then, there exists a polynomial $g_k(x) \in \mathbb{F}_2[x]$ of degree $2^n-n-2$ such that $\displaystyle{\frac{g_k(x)}{F(x)}}$ is a rational fraction representation of the $k$-shifted sequence $L^k \, \mathbf{s}$ for some $0 \leq k < N$. The polynomial $g_k(x)$ is the canonical generator of the Hamiltonian cycle $\mathcal{H}$ that corresponds to $\mathbf{s}$.
\end{theorem}
\begin{proof}
	We claim that  $\deg(g_k(x)) \geq 2^n-n-2$. Suppose, on the contrary, that there exists some $r$, with $0 \leq r < 2^n-1$, such that $\deg(g_r(x)) < 2^n-n-2$. Let 
	\[
	[n-r] := \begin{cases}
		n - r & \mbox{if } n \geq r,\\
		n-r + 2^n-1 & \mbox{if } n < r.
	\end{cases}
	\]
	Then $g(x) \, x^{[n-r]} = g(x) \, x^{2^n-1-r+n} \equiv g_r(x) \, x^n \Mod{F(x)}$. Therefore, 
	\begin{equation}\label{eq:gx1}
		g(x) \, x^{[n-r]} \equiv (0 \Mod{F(x)}) \Mod{x^n},
	\end{equation}
	which contradicts the fact that $0 \notin \Omega$.
	
	Next, we show the existence of $g_u(x)$, for some $0 \leq u < 2^n-1$, that satisfies $\deg(g(x)) = 2^n-n-2$ and $g_u(0)=1$. Since $\Omega = \{g_k(x) \Mod{x^n} : 0 \leq k < 2^n-1\}$, there exist some $t$, with $ 0 \leq t < 2^n-1$, such that 
	$g_t(x) \equiv x^{n-1} \Mod{x^n}$. If $g_t(x)= \widetilde{g}_{t}(x) \, x^n + x^{n-1}$ with 
	$\deg(\widetilde{g}_{t}(x)) < 2^n-n-2$, then 
	\[
	g(x) \, x^{2^n-1-t} \equiv \widetilde{g}_{t}(x) \, x^n + x^{n-1} \Mod{F(x)}.
	\]
	Since $\gcd(x,F(x))=1$, we multiply both sides by $x^{1-n}$. Let $u := (t+n-1) \Mod{(2^n-1)}$. Then 
	$g(x) \, x^{2^n-1-u} = g_u(x) \equiv \widetilde{g}_{t}(x) \, x + 1 \Mod{F(x)}$. Finally, we infer 
	\[
	\deg(g_u(x))= \deg(\widetilde{g}_{t}(x)) + 1 \leq 2^n-n-2.
	\]
	Thus, $\text{deg}(g_u(x))=2^n-n-2$. By Corollary \ref{cor:existence of g(x)} and Definition \ref{def:canonical_generator}, we confirm that $g_k(x)$ is precisely $c_{\mathcal{H}}(x)$ of $\mathbf{s}$.
\end{proof}

\begin{corollary}\label{cor:correspond_mdbs}
	For $n \geq 4$, let $\Gamma_n$ be given. Then the followings hold.
	\begin{enumerate}
		\item Distinct Hamiltonian cycles in $\Gamma_n$ correspond to inequivalent modified de Bruijn sequences.
		\item The number of distinct Hamiltonian cycles in $\Gamma_n$ is $2^{2^{n-1}-n}$.
	\end{enumerate}
\end{corollary}
\begin{proof}
	We prove the statements in their order of appearance.
	\begin{enumerate}
		\item Let $\mathcal{H}_i$ be a Hamiltonian cycle in $\Gamma_n$. By Theorem \ref{thm:correspondence}, $\mathcal{H}_i$ corresponds to a modified de Bruijn sequence $\mathbf{s}_i$. By Corollary \ref{cor:existence of g(x)} and Theorem \ref{thm:canonical_g(x)}, there exists a generator polynomial $g(x)$ of degree $2^n-n-2$ such that 
		\[\mathcal{H}_i = \{(x^i~g(x) \Mod{F(x)}) \Mod{x^n} : 0 \leq i < 2^n-1\}.\]
		Let $\mathcal{H}_j$ be a Hamiltonian cycle in $\Gamma_n$ which is distinct from  $\mathcal{H}_i$. Then there exists a polynomial $g'(x) \in \mathbb{F}_2[x]$ such that 
		\[\mathcal{H}_j= \{(x^i~g'(x) \Mod{F(x)}) \Mod{x^n} : 0 \leq i < 2^n-1\}.\]
		Since $ \mathcal{H}_i \neq \mathcal{H}_j$, it is clear that $g'(x) \neq g(x)$. We can then conclude that distinct Hamiltonian cycles correspond to inequivalent modified de Bruijn sequences.
		\item By Theorem \ref{thm:canonical_g(x)}, any modified de Bruijn sequence has a representation $\frac{g(x)}{F(x)}$, with $\deg(g(x))=2^n-2-n$. We know from Lemma \ref{teotan2} that the set 
		\[
		\mathcal{O} := \{(x^i~g(x) \Mod{F(x)}) \Mod{x^n} : 0 \leq i < 2^n-1\}
		\]
		is the set of all vertices in $\Gamma_n.$ If $a(x) \in \Omega$, then the two possible values of $b(x)=x \, a(x)$ are given in Equations (\ref{eq:double}) and (\ref{eq:complement}). Hence, the set $\mathcal{O}$ corresponds to a Hamiltonian cycle in  $\Gamma_n$. Thus, the number of Hamiltonian cycles in $\Gamma_n$ is equal to the number of modified de Bruijn sequences.
	\end{enumerate}
\end{proof}

Every Hamiltonian cycle $\mathcal{H}$ in $\Gamma_n$ has $c_{\mathcal{H}}(x)$ of degree $2^n-n-2$. We now show how to determine $c_{\mathcal{H}}(x)$ for a given $\mathcal{H}$.

\begin{theorem}\label{thm:find_g}
	If $\mathcal{H} $ is a Hamiltonian cycle  in $\Gamma_{n \geq 4}$, then its canonical generator $c_{\mathcal{H}}(x)$ can be determined by the following procedure.
	\begin{enumerate}
		\item Begin by writing $c_{\mathcal{H}}(x) := x^{2^n-n-2}+ \sum_{i=1}^{2^n-n-3} c_i \, x^i + 1$.
		\item Let $k=0,1,\ldots,2^n-2$ and $i = 2^n-2-n-k$. Determine each  $c_i$ by solving the congruence 
		$ x^{n+k} c_{\mathcal{H}}(x) \equiv h_{n+k} \Mod{F(x)} \Mod{x^n}$, where $h_n = \sum_{i=0}^{n-1} x^i \in \mathcal{H}$.
	\end{enumerate}
\end{theorem}

\begin{proof}
	Since $c_{\mathcal{H}}(x) := x^{2^n-2-n}+ \sum_{i=1}^{2^n-n-3} c_i \, x^i + 1$ and $h_n = \sum_{i=0}^{n-1} x^i$ satisfies
	\[
	h_n \equiv x^n c_{\mathcal{H}}(x) \Mod{F(x)} \Mod{x^n},
	\]
	we know that $h_{n+k}$ determines the value of $c_{2^n-2-n-k}$. Thus, all $c_i$s can be determined.
\end{proof}

\begin{example}\label{ex:findg}
	Let $\mathcal{H} \in \Gamma_4$ be given in terms of its successive vertices 
	\[
	(1, 13, 5, 10, 11, 9, 2, 4, 7, 14, 3, 6, 12, 8, 15).
	\]
	We have $\displaystyle{c_{\mathcal{H}}(x) = 
		x^{10} +\sum_{i=1}^{9} c_i \, x^i +1}$ and $\displaystyle{F(x)=\sum_{i=0}^{14} x^i}$. Hence, $x^4 \, c_{\mathcal{H}}(x) = x^3+x^2+x+1$, which is vertex $15$. Consequently, $(c_{\mathcal{H}}(x) \Mod{F(x)}) \Mod{x^4} = x+1$ is vertex $3$. 
	Hence, $c_{\mathcal{H}}(x) = x^{10} + c_9 \, x^9 + c_8 \, x^8 + c_7 \, x^7 + c_6 \, x^6 + c_5 \,x^5 + c_4 \, x^4 + x + 1$. We iterate the process and read the resulting polynomials in terms of their vertices in $\Gamma_4$. Since $x^5 \,c_{\mathcal{H}}(x)$ is vertex $1$, we obtain $c_9 = 0$. Since $x^6 \, c_{\mathcal{H}}(x)$ is vertex $13$, we infer $c_8=1$. We end up with $c_{\mathcal{H}}(x) = x^{10} + x^8 + x^5 + x + 1$ from Example~\ref{ex:genpoly}.
\end{example}

Finally, we are ready to determine the minimal polynomial of a modified binary de Bruijn sequence.
\begin{theorem}\label{thm:gcdmin}
	Let $F(x) := x^{2^n-2} + x^{2^n-3}+ \ldots+ x^2+x+1$. Let $\mathcal{H}$ be a Hamiltonian cycle in $\Gamma_n$ whose canonical generator is $c_{\mathcal{H}}(x)$. The reciprocal polynomial $f^{*}(x)$ of 
	\begin{equation}\label{eq:m}
		f(x):=\frac{F(x)}{d(x)} \mbox{, where } d(x) := \gcd(c_{\mathcal{H}}(x),F(x)), 	
	\end{equation}
	is the minimal polynomial $m_{\mathbf{s}}(x)$ of the modified de Bruijn sequence $\mathbf{s}$ that corresponds to the given $\mathcal{H}$.
\end{theorem}
\begin{proof}
	Since $c_{\mathcal{H}}(x)$ is the canonical generator of $\mathcal{H}$ in $\Gamma_n$, we know that 
	\[
	\Omega = \left\{(x^i \, c_{\mathcal{H}}(x) \Mod{F(x)}) \Mod{x^n} : 
	0 \leq i \leq 2^n-2 \right\}.
	\]
	Hence, $\displaystyle{\frac{c_{\mathcal{H}}(x)}{F(x)}}$ is a rational fraction representation of $\mathbf{s}$. If $d(x):= \gcd(c_{\mathcal{H}}(x),F(x))$, then 
	\[
	c_{\mathcal{H}}(x) := d(x) \, \widetilde{c}(x) \mbox{ and } F(x) = d(x) \, f(x).
	\]
	Thus, $\displaystyle{\frac{\widetilde{c}(x)}{f(x)}}$ is also a rational fraction representation of $\mathbf{s}$, whose minimal polynomial is $m_{\mathbf{s}}(x) = f^{*}(x)$.
\end{proof}

\begin{example}
	When $n=4$, the canonical polynomial 
	$c_{\mathcal{H}}(x)= x^{10} + x^7 + x^5 + x + 1$ generates the cycle of vertices $(1, 2, 11, 9, 13, 5, 10, 4, 7, 14, 3, 6, 12, 8,  15)$. The arcs forming the Hamiltonian path is
	\begin{multline*}
		(1,2), \, (2, 11), \, (11,9), \, (9,13), \, (13,5), (5,10), \, (10, 4), \\
		(4,7), \, (7,14), \, (14,3), (3,6), \, (6, 12), \, (12,8), \, (8,15)
	\end{multline*}
	and the corresponding sequence is $\mathbf{s}=(0,1,1,1,1,0,0,1,0,1,0,0,0,1,1)$. Since $d(x)=\gcd(F(x),c_{\mathcal{H}}(x)) = x^2 + x + 1$,
	\[
	m_{\mathbf{s}}(x) = f^{*}(x) = f(x) = \frac{F(x)}{d(x)} = x^{12} + x^9 + x^6 + x^3 + 1.
	\]
\end{example}

This work has, thus, provides a systematic method to determine the minimal polynomial of a modified binary de Bruijn sequence. As a concluding remark we highlight that if one can, for any $n \geq 3$, characterize the occasions for which $\gcd((c_{\mathcal{H}}(x), F(x)) = 1$, then we can confirm that there exist modified de Bruijn sequences with maximal complexity $2^n-2$. Computational evidences for small values of $n$ strongly suggest that most modified de Bruijn sequences have maximal complexity. Determining a closed formula for the number of such sequences is a worthy research challenge to solve.

\section*{Acknowledgements}
Musthofa is supported by the Indonesian Endowment Fund for Education, known by its abbreviation LPDP in Bahasa Indonesia, a full-ride scholarship from the Indonesian Ministry of Finance. Nanyang Technological University Grant 04INS000047C230GRT01 supports the research carried out by M.~F.~Ezerman.


\end{document}